%% file: main.tex
\documentclass[conference]{IEEEtran}
\IEEEoverridecommandlockouts
\usepackage{cite}
\usepackage{amsmath,amssymb,amsfonts}
\usepackage{amsthm}

\usepackage{booktabs,array}
\usepackage{balance}
\usepackage{pgfplots}
\pgfplotsset{compat=1.16}
\usetikzlibrary{arrows.meta,positioning,shapes.geometric,calc,fit,backgrounds}
\newtheorem{theorem}{Theorem}
\newtheorem{proposition}{Proposition}
\newtheorem{lemma}{Lemma}
\newtheorem{corollary}{Corollary}
\usepackage{algorithm}
\usepackage{algorithmic}

\usepackage{graphicx}
\usepackage{textcomp}
\usepackage{xcolor}
\usepackage{url}
\usepackage[hidelinks,bookmarks=true,bookmarksnumbered=true]{hyperref}
\hypersetup{pdftitle={When a High Score Is an Illusion: Certifying Genuine versus Repackaged Forecasting Skill},
pdfauthor={Pin Ni, Francesca Medda, Ramin Okhrati},
pdfsubject={Reference designs and inference for forecast comparisons},
pdfkeywords={forecast evaluation, rank statistics, reference design, residual covariance}}
\def\BibTeX{{\rm B\kern-.05em{\sc i\kern-.025em b}\kern-.08em
    T\kern-.1667em\lower.7ex\hbox{E}\kern-.125emX}}
\begin{document}
\bstctlcite{fullauthornames}
\title{When a High Score Is an Illusion: Certifying\\Genuine versus Repackaged Forecasting Skill}

\author{\IEEEauthorblockN{Pin Ni\IEEEauthorrefmark{1}\IEEEauthorrefmark{2},
Francesca Medda\IEEEauthorrefmark{1}, Ramin Okhrati\IEEEauthorrefmark{1}}
\IEEEauthorblockA{\IEEEauthorrefmark{1}University College London, London, United Kingdom\\
\IEEEauthorrefmark{2}University of Luxembourg, Luxembourg\\
$\{$pin.ni.21, f.medda, r.okhrati$\}$@ucl.ac.uk; ni.pin@uni.lu}
}

\maketitle

\begin{abstract}
\input{paper/0_abstract}
\end{abstract}


\section{Introduction}\label{sec1}
\input{paper/1_introduction}
\input{paper/fig_reference_certificate}

\section{Related Work}\label{sec2}
\input{paper/2_related_work}

\section{Methodology}\label{sec3}
\input{paper/3_method}

\section{Experiments}\label{sec4}
\input{paper/4_experiment}

\section{Results and Analysis}\label{sec5}
\input{paper/5_analysis}

\section{Conclusion}\label{sec6}
\input{paper/6_conclusion}

\appendix[Proofs]\label{app:proofs}
\input{paper/appendix}

\begin{footnotesize}
\bibliographystyle{IEEEtran}
\bibliography{ref}
\end{footnotesize}

\end{document}

%% file: paper/0_abstract.tex
Ranks depend on the observations used for comparison. Reusing those
observations can add association between forecast and outcome rank contrasts
even when the evaluated forecast and outcome stay fixed. We characterize assignments that
preserve association between specified population-rank contrasts, including
forecast rank minus baseline rank compared with outcome rank minus baseline
rank. Conditional on
independent training, whole trajectories are sampled independently from a
common law, with unrestricted dependence within each trajectory. The expected
score separates into its target and an explicit interaction between map
pairs. When reassigning references, we keep the learned maps, reference
law and coefficient row sums fixed. Zero weighted reference overlap for every map pair is
necessary and sufficient for preservation uniformly over permitted maps and
laws. An unbiased three-trajectory kernel estimates the interaction;
independent evaluation and validation provide finite-sample lower bounds.
Complete U-statistics estimate the same target directly when all draws can
be recombined. Sharp shared-baseline ranges, including ties, tighten both
constructions. In a Beijing air-quality archive, interaction accounts for
91.4\% to 94.5\% of seven learned forecasts' expected shared scores under
the empirical archive law. Separate results address category-fitting error
and temporal feedback. In a matched category-control null experiment,
distinct references reduce rejections from 402 to 41 out of 1,000 panels.

%% file: paper/1_introduction.tex
A forecast that repeats last week's sales can be useful. To study its
contribution relative to that baseline, we must first specify what the
comparison measures~\cite{fair1990comparing,hewamalage2023pitfalls}.
We study forecast rank minus baseline rank compared with outcome rank minus baseline rank.

Ranks count how often a value exceeds reference values, splitting ties equally.
Each reference provides forecast, outcome and baseline values.
A reference with a lower forecast and outcome but a higher baseline than
those being evaluated contributes positively to both contrasts.
Changing the reference can move both contrasts while the evaluated values stay fixed.
Which assignments preserve the intended association?

\paragraph{Reference designs that preserve the target}
We reassign reference observations while keeping the learned maps,
reference law and coefficient row sums fixed.
We condition on independent training and independent whole trajectories from
that law; within-trajectory dependence is unrestricted.
The expected score separates into its target and an explicit map-pair interaction.
We prove that zero weighted overlap for every map pair characterizes target
preservation uniformly over permitted maps and laws.

\paragraph{From design to a certificate}
An unbiased three-trajectory kernel estimates the interaction. Independent
evaluation and validation give finite-sample lower bounds.
Complete U-statistics directly estimate the same target when all draws can be
recombined. Sharp shared-baseline ranges tighten both constructions.
Full baseline adjustment is a separate target; category and temporal branches
retain their own conditions (Table~\ref{tab:operating}).
Figure~\ref{fig:reference-certificate} shows the inputs. Experiments distinguish
calibration, certification and subsequent forecast loss.

%% file: paper/fig_reference_certificate.tex
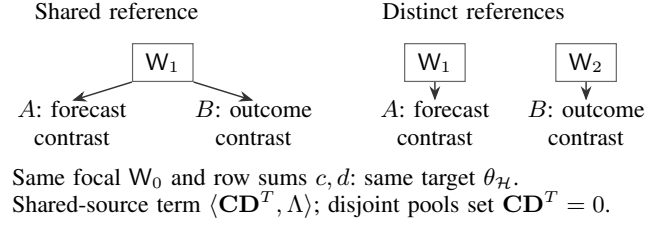
\begin{figure*}[t]\centering
\input{paper/fig_reference_certificate_tikz}
\caption{Independent source designs. (a) Aggregate validation uses exact category masses to bound learning bias; evaluation supplies the score and sampling radius. (b) Whole trajectories share one law and are independent given $\mathcal H$. Fixed maps and row sums preserve the fitted target when references are separated. Fresh triples estimate the interaction when sharing is retained.}
\label{fig:reference-certificate}
\end{figure*}

%% file: paper/fig_reference_certificate_tikz.tex
% Native vector diagram at publication size: IEEEtran default Times and CM math.
\begin{tikzpicture}[
 x=1mm,y=1mm,font=\fontsize{9}{10.2}\selectfont,>=Stealth,
 box/.style={draw=black!65,line width=.45pt,align=left,inner sep=1.1mm,anchor=north west},
 arrow/.style={->,line width=.55pt,draw=black!80},
 note/.style={align=left,inner sep=0pt,anchor=north west},
 ref/.style={draw=black!60,line width=.45pt,minimum width=8mm,minimum height=5mm,inner sep=0pt},
 val/.style={inner sep=0pt}]
\path[use as bounding box] (0,0) rectangle (174,-50);
\node[note] at (0,0) {\textbf{Training fixes learned maps and coefficients before validation and evaluation.}};
\draw[black!35,line width=.4pt] (87,-5) -- (87,-50);
\node[note] at (0,-5) {\textbf{(a) Category certificate}};
\node[note] at (91,-5) {\textbf{(b) Whole-trajectory design}};
\node[box,text width=36mm] (validate) at (0,-10)
 {\textbf{Validation $A,B$}\\Independent pairs\\Known category masses};
\node[box,text width=37mm] (evaluate) at (44,-10)
 {\textbf{Evaluation $E$}\\Independent groups};
\node[box,text width=36mm] (allowance) at (0,-27)
 {Learning allowance\\$\widehat B_{\rm agg}$: category errors};
\node[box,text width=37mm] (sampling) at (44,-27)
 {Corrected score $\bar D$\\Sampling radius $r_U$};
\draw[arrow] (validate.south) -- (validate.south |- allowance.north);
\draw[arrow] (evaluate.south) -- (evaluate.south |- sampling.north);
\node[box,text width=81mm] (cat) at (0,-41)
 {$L_\alpha=\bar D-\widehat B_{\rm agg}-r_U(\alpha)>0$\\Positive association beyond the declared categories};
\draw[arrow] (allowance.south) -- (allowance.south |- cat.north);
\draw[arrow] (sampling.south) -- (sampling.south |- cat.north);
\node[note,text width=83mm] at (91,-10)
 {$\mathsf W_i$ independent under one law given $\mathcal H$;\\dependence within each trajectory is unrestricted.};
\node[note] at (94,-19) {Shared reference};
\node[note] at (140,-19) {Distinct references};
\node[ref] (shared) at (111,-27) {$\mathsf W_1$};
\node[val,align=center] (sa) at (99,-35) {$A$: forecast\\contrast};
\node[val,align=center] (sb) at (123,-35) {$B$: outcome\\contrast};
\draw[arrow] (shared.south west) -- (sa.north);
\draw[arrow] (shared.south east) -- (sb.north);
\node[ref] (ra) at (147,-27) {$\mathsf W_1$};
\node[ref] (rb) at (167,-27) {$\mathsf W_2$};
\node[val,align=center] (da) at (147,-35) {$A$: forecast\\contrast};
\node[val,align=center] (db) at (167,-35) {$B$: outcome\\contrast};
\draw[arrow] (ra.south) -- (da.north);
\draw[arrow] (rb.south) -- (db.north);
\node[note,text width=83mm] at (91,-41)
 {Same focal $\mathsf W_0$ and row sums $c,d$: same target $\theta_{\mathcal H}$.\\Shared-source term $\langle \mathbf C\mathbf D^T,\Lambda\rangle$; disjoint pools set $\mathbf C\mathbf D^T=0$.};
\end{tikzpicture}

%% file: paper/2_related_work.tex
Forecast evaluation separates extra information from loss improvement
\cite{fair1990comparing,giacomini2006tests}. Clark--West adjusts
nested-model squared-loss comparisons~\cite{clark2007approximately};
sequential methods bound average predictable score differences~\cite{choe2024comparing}.
We target fitted rank association under independent trajectory sampling.

Each scalar component of our symmetric kernel is one twelfth of the
classical order-three Spearman kernel~\cite[Example~2.2]{leung2018testing}.
We characterize target-preserving assignments through map-pair interactions.
Preprocessing can already induce either bias sign~\cite{moscovich2022preprocessing}.
Aggregate validation specializes higher-order bias estimation
\cite{robins2017higher,liu2024falsification}: for fixed fits and known category
masses, separate validation streams bound the signed error product.
McGrath--Mukherjee studies tuning and splitting for optimal
conditional-covariance estimation~\cite{mcgrath2026tuning}.

GCM and cross-fitting require regression-error conditions
\cite{shah2020hardness,chernozhukov2018dml,chiang2022multiway};
kernel, partial-copula and empirically calibrated tests address conditional
independence~\cite{zhang2011kernel,petersen2021partial,pan2026eccit}.
Concentration and betting supply our sampling tools
\cite{maurer2009empirical,peel2010empirical,waudby2024betting}.
Gini covariance and semi-supervised U-statistics use additional forecast
information~\cite{schechtman1987gini,kim2025semi}; recursive adjustments precede
our explicit temporal source identities~\cite{so1999recursive,arellano1995another}.

%% file: paper/3_method.tex
\subsection{Target and interpretation}
Let $X$ be a forecast, $Y$ its outcome and $Z$ the declared controls.
Panel analyses use standardized within-period average ranks. Let $m_X=\mathbb E[X\mid Z]$, $m_Y=\mathbb E[Y\mid Z]$ and
$r_X=X-m_X$, $r_Y=Y-m_Y$. The full-adjustment reference target is
\begin{equation}\label{eq:target}
 \theta_\infty=\mathbb E[r_Xr_Y],\qquad
 H_0:\theta_\infty\le0\quad\text{against}\quad H_1:\theta_\infty>0.
\end{equation}
If $v_X=\mathbb E r_X^2>0$, the best augmentation $m_Y+t r_X$, $t\ge0$,
reduces squared-error risk by $(\theta_\infty)_+^2/v_X$:
expand its risk as $\mathbb E r_Y^2-2t\theta_\infty+t^2v_X$.
For any fitted $\widehat t$ and positive $\theta_\infty$, the gain is
$\theta_\infty^2/v_X-v_X(\widehat t-\theta_\infty/v_X)^2$.
Certification concerns positive residual association under the stated
sampling model. Zero covariance permits nonlinear dependence.

For population fits $a_X,a_Y$ in a declared additive class,
conditional centering removes both cross terms in the residual target:
\begin{equation}\label{eq:projection-target}
\begin{split}
 \theta_{\mathcal A}&:=\mathbb E[(X-a_X)(Y-a_Y)]\\
 &=\theta_\infty+\mathbb E[(m_X-a_X)(m_Y-a_Y)].
\end{split}
\end{equation}
The second term must be negligible to interpret this as $\theta_\infty$.
Table~\ref{tab:operating} separates the targets and required information.

\input{paper/operating_choices}
\input{paper/trajectory_design_main}
\input{paper/reference_rank_main}

\input{paper/fixed_peer_sampling}
\input{paper/feedback_identity}

\subsection{Observed-panel sensitivity and selection}
Panel screens multiply cross-fitted residuals from additive entity effects
and baseline bins. Ordered panels use Newey--West variance at the stated
Bartlett lag~\cite{neweywest1987}; ratings use cluster variance.
Valid inference needs a cluster CLT, consistent variance and negligible
fitted-product error~\cite{shah2020hardness,chiang2022multiway}.

An outcome-free rule abstains on exact copies, matching or reversed cluster
orders with ties, and constant ranks. Setting $p^*=1$ preserves valid
inputs because $p^*\ge p$, without repairing invalid ones.
Resolution screens combine a fixed bin ladder to cancel a prespecified
leading bias term; remaining bias and fitting error must be negligible
relative to the standard error. The bias-decay exponents $\beta=1,2$ and splines are fixed before
evaluation.

\begin{algorithm}[b]
\caption{Audit a fixed forecast family}\label{alg:audit}
\begin{algorithmic}[1]
\STATE Fix the target, full family and design (Table~\ref{tab:operating}) before evaluation.
\STATE For a certificate, use independent training to fix maps, coefficients, bounds and budgets. For a panel screen, record its fitting protocol.
\STATE Check source IDs, row sums and forecast availability; justify the branch's sampling assumptions.
\STATE Return the category bound~\eqref{eq:category-certificate}, the trajectory bound~\eqref{eq:trajectory-bound} or a labeled panel screen.
\STATE For $p$-value rules, set abstentions to $p^*=1$ and leave other $p$-values unchanged; apply BY for FDR.
\STATE For fixed-level bounds, test each of $K$ candidates at $\alpha/K$ for Bonferroni familywise control. Return the target and validity status.
\end{algorithmic}
\end{algorithm}
Algorithm~\ref{alg:audit} fixes the rule and family.
A valid positive lower bound \emph{certifies} its target at the declared level;
\emph{retained} denotes a BY selection, and a guard triggers \emph{abstention}.
A raw $p<\alpha$ is a \emph{nominal positive}.
Panel outputs remain screens unless their inference conditions hold.

\subsection{Multiple testing and selection}

\noindent\begin{minipage}{\columnwidth}
\begin{proposition}[Complete-family error control]\label{prop:fdr}
For fixed $K$ and super-uniform true-null $p$-values, BY uses the step-up
thresholds $j\alpha_{\rm FDR}/(K H_K)$, $H_K=\sum_{j=1}^K1/j$, and controls the false discovery rate (FDR)
at $\alpha_{\rm FDR}$ under arbitrary dependence. Fixed $K$ and asymptotically valid
inputs give asymptotic control.
\end{proposition}
\end{minipage}
\par\smallskip
All models remain in the family, including abstentions. Selecting the family
or audit rule after viewing results requires a separate conditional-validity
argument or independent confirmation.

%% file: paper/operating_choices.tex
\begin{table}[t]\centering\footnotesize
\caption{Targets and required information.}\label{tab:operating}
\setlength{\tabcolsep}{3pt}
\begin{tabular}{@{}>{\raggedright\arraybackslash}p{.45\columnwidth}>{\raggedright\arraybackslash}p{.51\columnwidth}@{}}\toprule
Target & Required information\\\midrule
$\theta_{\mathcal H}$: association between specified fitted rank contrasts & Given $\mathcal H$, fixed maps and coefficients; i.i.d.\ whole trajectories.\\\addlinespace[2pt]
$\vartheta$: rank association after category adjustment & Independent categorical peers and validated fits; exact masses for aggregate validation.\\\addlinespace[2pt]
$\theta_\infty$: association after full adjustment for $Z$ & Conditional means or a justified remaining-bias bound, and a valid sampling model.\\\addlinespace[2pt]
$\theta_{\mathcal A}$: association after additive adjustment & Additive population fits; observed-panel inference needs joint calibration.\\\bottomrule
\end{tabular}\par\smallskip
\begin{minipage}{\columnwidth}\footnotesize
Known forecast ranks permit direct kernels; a complete outcome archive permits census.
Category-fit training $H$ and trajectory-map training $\mathcal H$ define separate designs.\end{minipage}
\end{table}

%% file: paper/trajectory_design_main.tex
\subsection{Reference design for learned forecasts}
\label{sec:trajectory-design}
For the sales example, let $\phi=(\text{forecast},\text{baseline})$ and
$\psi=(\text{sales},\text{baseline})$. Row sums $c,d$ fix the fitted contrasts:
$(1,-1)$ subtracts the baseline's centered population midrank on each side.
This defines $\theta_{\mathcal H}$, whereas category adjustment and
$\theta_\infty$ remove conditional means. For baseline $U$ uniform on $\{-1,0,1\}$, forecast and sales $U^2$
give $\theta_{\mathcal H}=7/54$ with $c=d=(1,-1)$; full and
category-adjusted targets are zero for $Z=U$.

Write $a(v,v')=\mathbf1\{v'<v\}+\tfrac12\mathbf1\{v'=v\}$ for a midrank comparison.
Condition on information $\mathcal H$ fixing scalar maps $\phi_s,\psi_t$
and coefficient matrices $\mathbf C,\mathbf D$. Whole trajectories
$\mathsf W_0,\ldots,\mathsf W_{J_{\rm ref}}$ are conditionally i.i.d.;
dependence inside each trajectory is unrestricted. Define
$p_{sj}=a(\phi_s(\mathsf W_0),\phi_s(\mathsf W_j))-1/2$ and $q_{tj}$ analogously.
Set $A=\sum C_{sj}p_{sj}$, $B=\sum D_{tj}q_{tj}$,
$c=\mathbf C\mathbf1$, $d=\mathbf D\mathbf1$,
$m_s=\mathbb E[p_{s1}\mid\mathsf W_0,\mathcal H]$ and $n_t$ analogously.

\begin{proposition}[Conditional trajectory design]\label{prop:trajectory}
With $\boldsymbol\Omega=\mathbf C\mathbf D^T$ and
$\Lambda_{st}=\mathbb E[\operatorname{Cov}(p_{s1},q_{t1}\mid\mathsf W_0,\mathcal H)\mid\mathcal H]$,
\begin{equation}\label{eq:trajectory-identity}
 \mathbb E[AB\mid\mathcal H]
 =\underbrace{\mathbb E[(c^Tm)(d^Tn)\mid\mathcal H]}_{\theta_{\mathcal H}}
 +\langle\boldsymbol\Omega,\Lambda\rangle.
\end{equation}
For fixed $\mathbf C,\mathbf D$, preservation for every map and law holds iff
$\boldsymbol\Omega=0$ entrywise. On three fresh trajectories,
$Q=\Delta p^T\boldsymbol\Omega\Delta q/2$ is unbiased for the interaction,
where $\Delta p_s=p_{s1}-p_{s2}$ and $\Delta q_t=q_{t1}-q_{t2}$.
\end{proposition}
For a fixed map and law, nonzero $\boldsymbol\Omega$ can still be orthogonal to
$\Lambda$; it does not establish bias. When reassigning references, hold the maps, reference law androw sums $c,d$ fixed. Independent training
may learn both maps and coefficients. Same-sample adaptation, random rank
standardization and dependence across trajectories need separate justification.
Forecast maps must also use only information available when issued.

For example, take $\phi=\psi=(U,V)$ with independent fair Bernoulli variables
and $\mathbf C=\mathbf D=(1,-1)^T$. Every map and coefficient row is nonzero;
$\boldsymbol\Omega=\left(\begin{smallmatrix}1&-1\\-1&1\end{smallmatrix}\right)$
has entry sum zero but $\Lambda=I_2/16$ gives interaction $1/8$.

\paragraph{Observable correction and uncertainty}
Use $M\ge1$ independent evaluation blocks and $n\ge1$ independent validation
triples. Their means are $\bar S=M^{-1}\sum_{b=1}^M A_bB_b$ and
$\bar Q=n^{-1}\sum_{b=1}^n Q_b$. Entrywise norms give valid range-width bounds
$W_S=\|\mathbf C\|_1\|\mathbf D\|_1/2$ and $W_Q=\|\boldsymbol\Omega\|_1$.
For observations $x$ of size $m$, set $h(W,m,u)=W\sqrt{\log(1/u)/(2m)}$.
For $m\ge2$, with unbiased sample variance $s_x^2$, the preselected hybrid uses
\begin{equation}\label{eq:trajectory-hybrid}
\begin{split}
 e(x,W,u)&=s_x\sqrt{\frac{2\log(2/u)}m}+\frac{7W\log(2/u)}{3(m-1)},\\
 r(x,W,u)&=\min\{h(W,m,u/2),e(x,W,u/2)\},
\end{split}
\end{equation}
The pure-range rule sets $r=h$, also used when $m=1$.
For fixed $0<\delta<\alpha<1$.
\begin{equation}\label{eq:trajectory-bound}
 L_\alpha=\bar S-\bar Q-r(S,W_S,\alpha-\delta)-r(Q,W_Q,\delta)
\end{equation}
has conditional noncoverage at most $\alpha$ by classical concentration
\cite{hoeffding1963,maurer2009empirical}. Invert the preselected rule as
$p=\inf(\{\alpha\in(\delta,1):L_\alpha>0\}\cup\{1\})$.
The released inversion routine uses the range rule; the hybrid is evaluated at a fixed level.
A known zero range gives zero uncertainty radius. Total cost is
$N_{\rm tot}=M(J_{\rm ref}+1)+3n$ trajectories, plus training.
Columns zero in both matrices can be deleted without changing any score or
width. Let $g=1+J_{\rm active}$ count the focal and nonzero reference roles.
Minimizing the declared range radius $k_S/\sqrt M+k_Q/\sqrt n$,
$k_S,k_Q>0$, under $gM+3n\le N$ gives
$M/n=(3k_S/(gk_Q))^{2/3}$. The constants combine range widths and
log-levels, fixed before observing scores.

\paragraph{Direct estimation when all trajectories are available}
The complete distinct-reference U-statistic uses every draw to estimate the
same target. Its six-role kernel has range width at most
$W_U=\|c\|_1\|d\|_1/6$. With $q=\lfloor N_{\rm tot}/3\rfloor$,
the generic-width radius $h(W_U,q,\alpha)$ is less than
$h(W_S,M,\alpha-\delta)/\sqrt6$ for positive integer $M,n,J_{\rm ref}$
and $W_S>0$. This compares radii, not realized centers or power.
Thus the interaction correction diagnoses an existing design; complete U
provides a stronger range bound when the same raw draws can be recombined.

\paragraph{Sharper bounds for a shared baseline}
For $\phi=(h,b)$, $\psi=(Y,b)$ and $\mathbf C=\mathbf D=(1,-1)^T$,
the sharp ranges, including ties, are
$S\in[-1/4,1]$, $Q\in[-1/2,2]$ and $U\in[-1/6,1/3]$
for the six-role symmetric kernel.
Substituting widths $5/4$, $5/2$ and $1/2$ tightens both constructions.
A declared map identity $h=b$ gives zero kernels and widths;
equality on the observed sample alone does not establish that identity.

%% file: paper/reference_rank_main.tex
\subsection{Rank references and learning error}
To rank a product, the sales audit compares it with reference products.
Consider $M$ independent groups of $N\ge3$
i.i.d.\ raw triples $(V_i,W_i,Z_i)$ from a common law, with nuisance fits
$f,g\in[0,1]$ fixed by independent category-fit training $H$
(distinct from the trajectory-map training information $\mathcal H$).
Using the midrank comparison $a$, let
$a_{ij}=a(V_i,V_j)$ and $b_{ij}=a(W_i,W_j)$. Set $k=N-1$, $U_{V,i}=k^{-1}\sum_{j\ne i}a_{ij}$
and $U_{W,i}=k^{-1}\sum_{j\ne i}b_{ij}$.
Let $F_V(v)=\Pr(V<v)+\tfrac12\Pr(V=v)$ and define $F_W$ analogously.
Write $R_V=F_V(V)$,
$R_W=F_W(W)$, $m_V=\mathbb E[R_V\mid Z,H]$ and $m_W$ analogously.
The target is $\vartheta=\mathbb E[(R_V-m_V)(R_W-m_W)\mid H]$,
separate from standardized panel ranks.
The shared-reference score $C_i=(U_{V,i}-f(Z_i))(U_{W,i}-g(Z_i))$ becomes
\begin{equation}\label{eq:reference-correction-main}
 D_i=C_i+\frac{U_{V,i}U_{W,i}-k^{-1}\sum_{j\ne i}a_{ij}b_{ij}}{k-1}.
\end{equation}
This uses distinct reference indices without discarding observations.
Given fitted means, the correction needs only two marginal comparison sums
and one joint-product sum per row. Sorting and cumulative counts evaluate
all $N$ scores in $O(N\log N)$ time and $O(N)$ working memory, including ties.

\begin{proposition}[Reference-reuse decomposition]\label{prop:reference-main}
Write $b_H=\mathbb E[(m_V-f)(m_W-g)\mid H]$ and
$\Gamma=\mathbb E[a_{12}b_{12}-F_V(V_1)F_W(W_1)\mid H]$.
Then, including ties,
\begin{equation}\label{eq:reference-decomposition-main}
 \mathbb E[C_i\mid H]=\vartheta+b_H+\Gamma/k,\qquad
 \mathbb E[D_i\mid H]=\vartheta+b_H.
\end{equation}
\end{proposition}
Distinct references remove reuse bias; either conditional-mean fit being exact removes $b_H$.
Cauchy--Schwarz gives $|b_H|\le\varepsilon_f\varepsilon_g$,
where $\varepsilon_f^2=\mathbb E[(m_V-f)^2\mid H]$ and similarly for $g$.
The remaining fitting error depends on their product.

Set $r_V=R_V-m_V$, $r_W=R_W-m_W$ and, for an independent copy $O'$,
\begin{equation}\label{eq:reference-projection}
\begin{split}
 h_V(v)&=\mathbb E[(a(V',v)-m_V(Z'))r_W'],\\
 h_W(w)&=\mathbb E[r_V'(a(W',w)-m_W(Z'))],\\
 \psi&=r_Vr_W+h_V(V)+h_W(W)-3\vartheta,
 \quad \sigma_\psi^2=\mathbb E\psi^2.
\end{split}
\end{equation}

\begin{theorem}[Covariance boundary and reference reuse]\label{thm:reference-main}
For $M,N\to\infty$, a fixed common law, exact $f=m_V$, $g=m_W$
and $\sigma_\psi^2>0$, let $A_j$ be group $j$'s mean ($A=C,D$),
$\bar A=M^{-1}\sum_j A_j$,
$\widehat v_A=(M-1)^{-1}\sum_j(A_j-\bar A)^2$ and
$T_A=\sqrt M\bar A/\sqrt{\widehat v_A}$, set to zero at zero variance.
Then
\begin{equation}\label{eq:general-reference-limit}
\begin{split}
 \sqrt{MN}(\bar D-\vartheta)&\Rightarrow\mathcal N(0,\sigma_\psi^2),\\
 \sqrt{MN}\{\bar C-\vartheta-\Gamma/(N-1)\}
 &\Rightarrow\mathcal N(0,\sigma_\psi^2),\\
 N\widehat v_A&\to_p\sigma_\psi^2.
\end{split}
\end{equation}
At $\vartheta=0$, $T_D\Rightarrow\mathcal N(0,1)$. At the same boundary, if $M/N\to\lambda<\infty$, then
$T_C\Rightarrow\mathcal N(\Gamma\sqrt\lambda/\sigma_\psi,1)$;
if $\Gamma>0$ and $M/N\to\infty$, its one-sided rejection probability
tends to one. At fixed negative or positive $\vartheta$, the corrected
rejection probability tends to zero or one, respectively.
Under $V\perp W\mid Z$, $h_V=h_W=0$ and
$\sigma_\psi^2=\mathbb E[r_V^2r_W^2]$.
\end{theorem}
The group variance accounts for focal and both reference roles without estimating $h_V,h_W$.
At fixed $N$ with positive group variance, the ordinary group CLT applies.
\par\smallskip
\noindent\begin{minipage}{\columnwidth}
\begin{corollary}[Estimated conditional means]\label{cor:estimated-reference}
Replace the oracle means in Theorem~\ref{thm:reference-main} by
fits $f,g\in[0,1]$ trained on a common sample independent of all evaluation groups. If
$\varepsilon_f+\varepsilon_g=o_p(1)$ and
$\sqrt{MN}\,\varepsilon_f\varepsilon_g=o_p(1)$,
all its studentized limits remain valid.
\end{corollary}
\end{minipage}
\par\smallskip

\input{paper/reference_certificate_main}

%% file: paper/reference_certificate_main.tex
\paragraph{Learning and sampling allowances}
For $C_Z$ categories, $m_{\rm val}$ validation pairs $(O,O')$ are independent of training
and evaluation. Focal $Z,V,W$ and reference $V',W'$ supply comparisons
with conditional means $m_V(c),m_W(c)$. With $n_c$ focal rows, their
Hoeffding intervals have radius $\sqrt{\log(8C_Z/\delta)/(2n_c)}$ and
are intersected with $[0,1]$;
empty cells receive $[0,1]$.
Let $q_c^+,q_c^-$ be the largest and smallest of the four endpoint products
$(x-f(c))(y-g(c))$ over these two intervals. Obtain binomial upper limits
$u_c$ for category probabilities, each with error $\delta/(2C_Z)$
\cite{clopper1934use}. Define
\begin{equation}\label{eq:category-budget}
 \widehat B_+=\max_{p\in\mathcal P}\sum_c p_cq_c^+,
 \qquad \widehat B_-=\min_{p\in\mathcal P}\sum_c p_cq_c^-.
\end{equation}
Here $\mathcal P=\{p\ge0:\mathbf1^Tp=1,\ p\le u\}$.
Sorted coefficient allocation gives the extrema, including negative terms.
\par\begingroup\interlinepenalty=10000
With probability at least $1-\delta$,
$\widehat B_-\le b_H\le\widehat B_+$. The signed upper allowance can be negative and is no larger than
$\widehat B_{\rm abs}$ from the largest absolute endpoint errors.\par\endgroup

\input{paper/aggregate_validation_main}
The certificate combines this validation allowance with an independent
evaluation-sample bound.
\begin{proposition}[Observable reference certificate]\label{prop:category-certificate}
Under the independent common-law reference design, let $\widehat B$ be an
independently validated upper bound on $b_H$ with failure probability $\delta$.
Let $R$ be the range
of $(a-f(c))(b-g(c))$ over $a,b\in[0,1]$ and all categories.
Fix disjoint triples before inspecting their values. Let $s^2$ be the
sample variance of their six-role symmetric kernels, with denominator
$J-1$, where $J=M\lfloor N/3\rfloor\ge2$.
For $0<\delta<\alpha<1$, write $x=\log\{2/(\alpha-\delta)\}$ and
\begin{align}
 c_J&=\frac{2R}{\sqrt{J(J-1)}}+\frac{R}{3J},\nonumber\\
 r_U(\alpha)&=\min\left\{R\sqrt{\frac{x}{2J}},\,
 s\sqrt{\frac{2x}{J}}+c_Jx\right\},\label{eq:variance-radius}\\
 L_\alpha&=\bar D-\widehat B-r_U(\alpha).\label{eq:category-certificate}
\end{align}
Then $\Pr(L_\alpha>\vartheta)\le\alpha$, including ties and zero
kernel variance. The mean $\bar D$ uses all distinct triples; only its
variance estimate uses the fixed disjoint triples.
\end{proposition}
This combines classical U-statistic concentration and variance bounds
\cite{maurer2009empirical,peel2010empirical} with either preselected allowance.
Use $\widehat B=\widehat B_+$ or $\widehat B_{\rm agg}$; taking their random
minimum requires a joint error budget. Define $p_{\rm cert}=\inf(\{\delta<\alpha<1:L_\alpha>0\}\cup\{1\})$. It is super-uniform
under $\vartheta\le0$. A range-only alternative replaces $r_U$ by
$R\sqrt{\log\{1/(\alpha-\delta)\}/(2J)}$ and also allows $J=1$.
Choose the construction before evaluation. For BY at level $\alpha_{\rm FDR}$,
$\delta=\rho\alpha_{\rm FDR}/(K H_K)$, $0<\rho<1$, places the $p$-value floor below
the first family threshold; $H_K=\sum_{j=1}^K1/j$.

\paragraph{Sufficient evaluation budget}
Apply Lemma~\ref{lem:aggregate-allowance} at $\delta$ to each sign of
stream $A$, including its corner bounds. The resulting interval
$[B_-,B_+]$ has $B_+=\widehat B_{\rm agg}$ and joint coverage
at least $1-2\delta$. On this event, put $W=B_+-B_-$. For
$\vartheta\ge\vartheta_0>W$ and $0<\eta<1$, the range-only rule
has conditional power at least $1-\eta$ if
\begin{equation}\label{eq:aggregate-power-cost}
 J>\frac{R^2}{2(\vartheta_0-W)^2}
 \left[\sqrt{\log\frac1{\alpha-\delta}}+
       \sqrt{\log\frac1\eta}\right]^2.
\end{equation}
Here $\alpha$ is a marginal cutoff; $\alpha=\alpha_{\rm FDR}/(K H_K)$ suffices for
first-threshold BY retention. Cost:
$n_{\rm train}+2m_{\rm val}+MN$ raw draws. Choose $M,N$ before fresh
evaluation; $\vartheta_0$ is a planning alternative, not an estimated effect.
Unconditional power also depends on the random width $W$.

%% file: paper/aggregate_validation_main.tex
\paragraph{Validate the aggregate when category masses are known}
The learning bias is one weighted sum. Suppose exact category probabilities
$p_c$ are known from control metadata or design; complete forecast or
outcome rank metadata are not required. Split independent validation pairs into two
streams before inspecting values. In category $c$, average
$a(V,V')-f(c)$ in stream $A$ and $a(W,W')-g(c)$ in stream $B$, giving
$\widehat a_c,\widehat b_c$ and counts $n_{Ac},n_{Bc}$.
The total $m_{\rm val}=\sum_c(n_{Ac}+n_{Bc})$ pairs cost $2m_{\rm val}$ raw observations.
Their product estimates the bias to subtract (Fig.~\ref{fig:reference-certificate}).

\begin{lemma}[Aggregate validation allowance]\label{lem:aggregate-allowance}
Condition on training and all focal validation category assignments. For
$\mathcal S=\{c:n_{Ac}n_{Bc}>0\}$, let $x=\log(3/\delta)$,
$d_c=p_c/(4\sqrt{n_{Ac}n_{Bc}})$ and
\begin{align}
 r_A^2&=\frac{x}{2}\sum_{c\in\mathcal S}\frac{p_c^2\widehat b_c^2}{n_{Ac}},
 &r_B^2&=\frac{x}{2}\sum_{c\in\mathcal S}\frac{p_c^2\widehat a_c^2}{n_{Bc}},\nonumber\\
 r_{AB}&=\sqrt{2x\sum_{c\in\mathcal S}d_c^2}+x\max_{c\in\mathcal S}d_c.\label{eq:aggregate-radii}
\end{align}
Take empty sums and maxima as zero. Let $q_c^{\rm full}$ maximize
$(u-f(c))(v-g(c))$ over $u,v\in\{0,1\}$. Then
\begin{equation}\label{eq:aggregate-budget}
 \widehat B_{\rm agg}=\sum_{c\in\mathcal S}p_c\widehat a_c\widehat b_c
 +r_A+r_B+r_{AB}+\sum_{c\notin\mathcal S}p_cq_c^{\rm full}
\end{equation}
satisfies $\Pr(b_H>\widehat B_{\rm agg})\le\delta$, including ties.
\end{lemma}
The radii bound the two linear errors and their product, respectively.
This is a finite rank-target specialization of classical second-order
bias estimation~\cite{robins2017higher,liu2024falsification} and concentration.
It replaces simultaneous mean rectangles by one observable aggregate bound.
Exact category masses are essential; plugging in empirical frequencies
requires another uncertainty argument.

%% file: paper/fixed_peer_sampling.tex
\paragraph{Peer sampling model}
The reference correction depends on the sampling design.
Condition on a fixed set of $N$ entity attributes, assume that entity pairs
are independent across entities and that the two channels are independent
within every entity, and let $f_i=\mathbb E[U_{V,i}]$ and
$g_i=\mathbb E[U_{W,i}]$ be the exact entity means. Set
$p_{ij}=\mathbb E[a_{ij}]$, $q_{ij}=\mathbb E[b_{ij}]$; covariance weights
each of the $k$ peers equally. Then
\begin{equation}\label{eq:fixed-peers-main}
\mathbb E[C_i]=0,\qquad
\mathbb E[D_i]=-\frac{\operatorname{Cov}_{j\ne i}(p_{ij},q_{ij})}{k-1}.
\end{equation}
Distinct references
remove sampling bias for i.i.d.\ resampled peers but can introduce it for
fixed heterogeneous peers. Seeing the same entities repeatedly does not determine which sampling model is appropriate.

%% file: paper/feedback_identity.tex
\subsection{Feedback through entity effects}
\label{sec:feedback-main}
Complementary fitting lets earlier outcomes enter evaluated forecasts and
evaluated outcomes enter later forecasts. The following identity isolates
these routes and the product of fitting errors.

\begin{proposition}[Feedback under temporal splitting]\label{prop:feedback-main}
For $T$ observations in contiguous folds with deterministic weights, let
$y_r=\alpha+\varepsilon_r$ and $x_r=\sum_{j<r}w_{rj}y_j+\eta_r$.
The square-integrable shocks have mean zero and variance $\sigma^2>0$;
they are mutually uncorrelated and uncorrelated with $\alpha$ and every $\eta_s$.
Set $w_{sj}=0$ for $j\ge s$. For evaluation fold $I$, let $E,L$ be all
earlier and later rows, $O=E\cup L$ and $m=|O|>0$.
\par\begingroup\interlinepenalty=10000
With both entity means fitted on $O$, set
$S_I=\sum_{r\in I}(x_r-\bar x_O)(y_r-\bar y_O)$. Then
\begin{equation}\label{eq:feedback-main}
\begin{split}
\frac{\mathbb E S_I}{\sigma^2}
={}&-\frac{\sum_{r\in I,j\in E}w_{rj}}{m}
     -\frac{\sum_{s\in L,r\in I}w_{sr}}{m}\\
&+\frac{|I|}{m^2}\sum_{s,j\in O}w_{sj}.
\end{split}
\end{equation}
\endgroup
If $E,L$ are nonempty, then
$\mathbb E[(x_r-\bar x_E)(y_r-\bar y_L)]=0$ for every $r\in I$.
This latter result also holds for any square-integrable forecast measurable
from $\alpha$, the forecast noises and outcomes before $r$, provided
$\mathbb E[\varepsilon_s\mid\alpha,\eta_1,\ldots,\eta_T,
\varepsilon_1,\ldots,\varepsilon_{s-1}]=0$ for every $s$.
\end{proposition}

The first two terms oppose nonnegative forecast weights; the third can reverse their sign.
Temporal splitting fits forecast effects earlier and outcome effects later,
adapting recursive adjustment and forward orthogonal deviations~\cite{so1999recursive,arellano1995another}.
This centers raw outcomes with entity means, not the experimental rank/baseline-adjusted panels.
The companion preserves the equal-fold closed form and its derivation.

\begin{corollary}[Independent-entity inference]\label{cor:directional}
Let $S_g$ be entity $g$'s mean directional product over fixed eligible rows
with nonempty earlier and later folds. Suppose entities, including their
forecast noises, are independent and satisfy Proposition~\ref{prop:feedback-main}'s
centering conditions. If $\sup_g\mathbb E|S_g|^4<\infty$ and
$G^{-1}\sum_{g=1}^G\operatorname{Var}(S_g)\to v>0$, then
$\sqrt G\,\bar S/\widehat{\operatorname{sd}}(S)\Rightarrow\mathcal N(0,1)$.
\end{corollary}
This raw-score test allows fixed observation counts per entity; centering
comes from the directional design, without consistent entity-effect fits.

The companion treats explicit lagged-source models;
Proposition~\ref{prop:trajectory} instead permits general fixed trajectory maps.

%% file: paper/4_experiment.tex
\subsection{Which parts of the audit change the evidence?}
Figure~\ref{fig:peer_sampling} contrasts resampled and fixed peers: shared scores center for fixed peers; distinct scores depend on effect alignment.

\input{paper/peer_feedback_evidence}

\input{paper/fig_certificate_efficiency}

\paragraph{Reference reuse and false positives}
In Fig.~\ref{fig:peer_sampling}(b), $Z\sim\mathrm{Bernoulli}(1/2)$ and
$V,W\mid Z$ are independent $\mathrm{Bernoulli}(0.2+0.6Z)$, giving $\vartheta=0$
and $\Gamma=0.0225$. Shared and distinct references reject $402$ and $41$
of $1{,}000$ paired panels with $8{,}192$ training rows; the corresponding
counts with $32$ training rows are $677$ and $527$.

A dependent zero-covariance example rejects $143/3{,}000$ with the full
scale and $482/3{,}000$ with the focal-only scale.

\paragraph{Added value of signed learning allowances}
The absolute/signed $\times$ range/variance ablation pairs $100$ inspected settings, each with $1{,}000$ repetitions.
At $30{,}976$ observations and eight categories
with $\vartheta\approx0.0137$, both range counts are zero; absolute and signed
variance bounds reject $274$ and $435$ times. The paired increase is $16.1$ percentage points (pointwise $95\%$ interval $[13.1,18.9]$). Figure~\ref{fig:certificate-efficiency}(a) fixes sampling precision across four budgets. All $160$ nonpositive method cells record zero rejections.

\paragraph{Comparison with classical sampling methods}
At that eight-category budget,
independent-triple and pooled-U variance counts are $433$ and $461$.
\par\begingroup\interlinepenalty=10000
A separate $300$-replication paired comparison at this budget and signed allowance uses
a fixed $64$-component betting mixture~\cite{waudby2024betting}, with
$p=\min(1,\delta+1/E)$. Betting detects $296/300$ versus $130/300$
for the variance rule: a paired gain of $55.3$ percentage points
($95\%$ interval $[47.2,61.8]$). Both null counts are $0/300$
(pointwise interval $[0,0.0123]$).\par\endgroup
\paragraph{Aggregate validation precision}
Figure~\ref{fig:certificate-efficiency}(b) uses $72$ settings, $2{,}000$ repetitions and $\delta=0.05$. \begingroup\interlinepenalty=10000
With $32$ balanced categories, $8{,}192$ validation pairs and exact fits,
median bias slack falls from $0.0185$ to $0.00260$.\par\endgroup
 Rectangles give smaller median slack in $10/72$ settings, including all
six two-category cases with opposite fitting errors. All $144$ method cells record zero noncoverage. These zero counts do not imply zero risk.

\paragraph{Temporal feedback}
In the raw-score null, gapping reverses complementary feedback from
$-0.0451$ to $0.0655$ and rejections from $0$ to $458/500$.
Directional expectation is zero; its worst tested rate is $44/500$
(Wilson $95\%$ interval $[0.066,0.116]$).

\paragraph{Complete-family error control}
BY tests $10$ or $11$ raw-score forecasts on $25$--$400$ independent
entities with $20$ observations: $1{,}000$ families, three positive controls
and an independent $5{,}000$-draw null bank per cell.
At $400$ entities and $11$ models, gapped complementary FDR falls from
$0.928$ with a Gaussian reference to $0.015$ after simulator calibration.
For three effects of $0.06$, calibrated power is $0.666$ for directional
blocks and $0.752$ for the classical orthogonal moment~\cite{arellano1995another}.
Both target the same constant effect; classical weights change the target
under heterogeneity. Calibration requires the null bank; the companion's
``Complete-family error and power'' table specifies the moments.

Additive-class misspecification gives $300/300$ nominal positives, without
establishing a zero standardized-rank target. In a near-copy null, linear-plus-spline
and $32$-bin adjustment reject $46$ and $982$ of $1{,}000$, respectively, under the same
$t_{24}$ and degrees-of-freedom correction.

\subsection{Public forecast protocol}
The $41$ forecasts use five contiguous folds, additive entity/baseline-bin controls, $\beta=1,2$ on $\{8,12,16,24,32\}$, and splines.
Undefined ranks, negligible residuals, standard error $\le10^{-10}$ or
observed copies receive $p^*=1$. Directional comparisons fix bins and copy flags; middle support uses folds $2$--$4$. Empty training sides use pooled complementary means.

Applications are Electricity~\cite{trindade2015electricity} ($11$ models),
M5~\cite{makridakis2022m5}, MovieLens-25M~\cite{harper2015movielens}
and Open Source Asset Pricing~\cite{chen2022open} ($10$ each).
Their baselines are seven-day demand, last-week sales, training
item mean and compounded past-year return. Heteroskedasticity and autocorrelation consistent (HAC) lags are $14$ daily,
$2$ weekly and $12$ monthly; ratings use user clustering.
The companion's ``Directional fitting support by task'' table reports fallback and unseen-entity shares, outside the temporal centering result.

%% file: paper/peer_feedback_evidence.tex
\begin{figure*}[t]\centering
\includegraphics[width=7in]{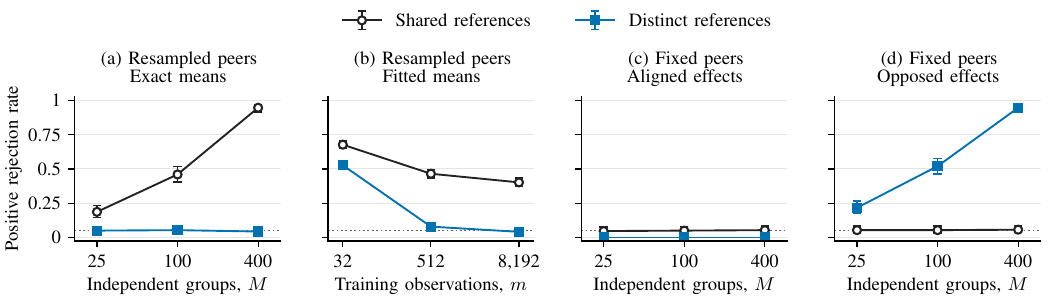}
\caption{Reference correction depends on sampling and mean estimation.
Positive null rejection rates with pointwise Wilson $95\%$ intervals;
the dotted line is $0.05$. Panels (a), (c), (d) use exact means, $N=32$
and $300$ replications. Panel (b) uses estimated means, $N=64$, $M=400$
and $1{,}000$ replications at each training size. Shared and distinct
references use the same evaluation panels. Fixed-peer effects are aligned
or opposed between the two channels.}
\label{fig:peer_sampling}
\end{figure*}

%% file: paper/fig_certificate_efficiency.tex
\begin{figure*}[t]
\centering\includegraphics[width=7in]{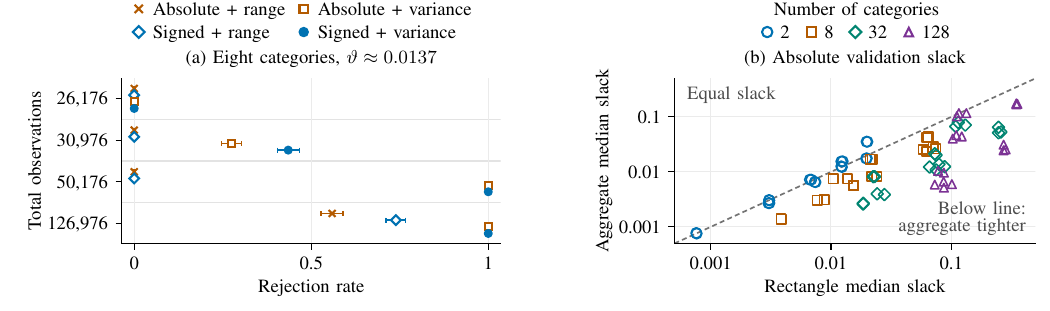}
\caption{Learning allowances and sampling precision.
(a) Eight-category learned-fit ablation at all four budgets: $1{,}000$
repetitions, pointwise exact binomial $95\%$ intervals, and small vertical
method offsets. (b) Absolute validation slack (upper bound minus true bias):
$72$ settings, $2{,}000$ repetitions each. Aggregation is tighter below equality ($62$ settings); rectangles above ($10$).
Slack alone does not determine the final association decision.}
\label{fig:certificate-efficiency}
\end{figure*}

%% file: paper/5_analysis.tex
\subsection{Same forecasts, different screening decisions}
On $68{,}000$ fixed middle-fold rows over $17$ weeks, retail price/calendar
moves from $6.41$ to $-0.82$. Croston SBA moves from $-19.43$ to $6.11$.
Retention gains three models and loses one.
Residual--fit cross terms and the product of fit changes decompose the shifts without identifying their causal feedback share.

All-fold complementary fitting retains $10/11$ electricity, $7/10$ retail,
$6/10$ ratings and $0/10$ asset-pricing models under both exponents and splines.
The ratings baseline has mean absolute error (MAE) $0.693$ and abstains;
alternating least squares has MAE $0.719$ and is retained across all four
fold/support choices. Ranks, fitting and dependence correction lack joint calibration.

\subsection{From residual evidence to forecast decisions}
Electricity boosting trained on $67{,}392$ earlier rows gives
$L_{0.05}=0.0705$ for target $0.0782$ in a $70{,}080$-row archive.
All seven bounds certify the seasonal forecast. A category-only copy
has zero target and lower bound $-0.00794$. Four categories cross weekend
status with zero/positive seasonal forecasts; $32{,}768$ draws revisit
$26{,}140$ existing rows.

Appliances energy~\cite{candanedo2017appliances} and Metro traffic
\cite{hogue2019metro} protocols were locally fixed before first download.
A midrank repair changed six copy decisions, preserving models and windows.
Each task has eight candidates and chronological training, calibration,
selection and confirmation. Gates cross $16$ training-forecast quantile bins
with weekend status, using $32{,}768$ selection-archive draws and
eight-candidate BY. Known forecast metadata permit classical
covariance kernels~\cite{schechtman1987gini} and pairs without learned means.

\input{paper/tab_public_confirmation}
With training-scale clipping (Table~\ref{tab:public-confirmation}), known-forecast
U, pairs and ungated selection choose ExtraTrees: MSE improves $16.95\%$
and $73.87\%$ over Ridge, with zero gain from gating. Reference allowances
exceed the associations, causing Ridge fallback.
\paragraph{Development diagnosis with category metadata}
\begingroup\interlinepenalty=10000
On four inspected eight-candidate families with matched metadata, fits
and $32{,}768$ draws, aggregate and fully averaged second-order allowances
both retain $13$ with betting and $9$ with pooled U. Rectangles with
betting retain $3$.\par\endgroup

\paragraph{Chronological news confirmation}
News Popularity~\cite{news2018data} fixed six candidates before source access.
On $1{,}110$ articles, five of nine gates select ungated boosting augmentation;
four return Ridge. The augmentation's clipped loss falls $0.158\%$ and raw
MSE rises $12.82\%$ versus Ridge.
No gate improves ungated clipped loss.

\paragraph{Learned trajectories and actual decisions}
Beijing PM2.5~\cite{chen2015beijing} forecasts issue at 14:00 for 15:00.
Trajectories contain the target, outcome lags $1,2,3,6,12,24$ hours
(lag 1 is the 14:00 outcome), issue-hour weather/wind and calendar fields.
Issue-time measurements are assumed available. Eight models train on
$23{,}006$ hours in 2010--2012;
selection uses $342$ daily trajectories in 2013 and evaluation $350$ in 2014.
For forecast $h$ and persistence $b$, $\phi=(h,b)$, $\psi=(Y,b)$ and
$\mathbf C=\mathbf D=\left(\begin{smallmatrix}1&0\\-1&0\end{smallmatrix}\right)$.
With-replacement draws from these $342$ trajectories implement
Proposition~\ref{prop:trajectory} for the empirical archive law:
$M=16{,}384$, $J_{\rm ref}=2$, $n=8{,}192$ with Bonferroni level $0.05/8$.
The $73{,}728$ draws revisit $2{,}069$ hourly rows, with no new labels.
\input{paper/fig_trajectory_comparison}
\input{paper/tab_trajectory_decomposition}
In this inspected-data comparison, shared-baseline ranges tighten all four methods
(Table~\ref{tab:trajectory-decomposition}). For Boosting 31, the exact
interaction is $0.00906$ of the $0.00991$ expected shared score ($91.4\%$).
Removing the unused column frees $16{,}384$ draw positions.
Reallocation within $73{,}728$ draws gives
$M=13{,}086,n=15{,}852$.
At that allocation, the sharper ranges reduce its hybrid radius from
\mbox{$0.0121$ to $0.00959$ ($20.9\%$)}, with unchanged center.

Neither shared-reference rule certifies a candidate. Triple hybrid certifies $3/8$; pooled variance certifies $4/8$, adding Boosting 7
($L=1.13\times10^{-5}$). Selected forecasts are unchanged. Pure-range
rules certify none.
Census takes
$0.010$ seconds versus pooled U $0.438$ for Boosting 31.

Shared-reference correction selects persistence (MSE $288.32$);
triple U, pooled U and ungated selection choose boosting (MSE $295.94$).
The loss gain is $7.62$ ($2.57\%$), with subsequent descriptive paired
four-week bootstrap interval $[-36.48,51.38]$ in MSE units.
Gating adds zero gain over choosing persistence directly.
An origin/target-hour mismatch was repaired after inspecting outcomes.
Both cohorts and their results are released.

\paragraph{Prespecified finite-law checks}
Two pre-simulation protocols test seven IID laws and two budgets, with
four candidates per family: the baseline, a forecast, its reversal and a constant.
Three methods each use two bounds, giving $7\times2\times3\times2=84$ cells
with $300$ repetitions each. A nonidentity forecast has zero target
and interaction $1/32$. All cells have
$0/300$ family noncoverage and false certifications
(pointwise two-sided $95\%$ interval $[0,0.0122]$). For target $1/135$ at $12{,}288$
draws, both U hybrid rules certify $300/300$ versus shared correction's
$0/300$; at $3{,}072$ draws all three certify $0/300$.
The companion gives all results under common structural ranges.

%% file: paper/tab_public_confirmation.tex
\begin{table}[t]\centering\small
\caption{Confirmation mean squared error.}
\label{tab:public-confirmation}
\setlength{\tabcolsep}{4pt}
\begin{tabular}{@{}lrr@{}}\toprule
 & \multicolumn{2}{c}{MSE ($\times10^{-3}$)}\\
Rule & Energy & Traffic\\\midrule
Ridge baseline & 20.80 & 5.252\\
Reference U / betting & 20.80 & 5.252\\
Ungated augmentation & 17.28 & 1.372\\
Known-forecast U / pairs & 17.28 & 1.372\\
Direct-loss bound & 20.80 & 1.372\\
Convex combination & 17.28 & 1.404\\
\bottomrule
\end{tabular}\par\smallskip
\begin{minipage}{\columnwidth}\footnotesize Confirmation MSE after training-scale clipping; lower is better.
Primary Ridge baseline, $1{,}008$ energy and $8{,}713$ traffic hours.
Reference rules equal Ridge; known-forecast rules equal ungated predictions.
Reference rules use the original rectangular validation protocol.\end{minipage}
\end{table}

%% file: paper/fig_trajectory_comparison.tex
\begin{figure*}[t]\centering
\includegraphics[width=7.1in]{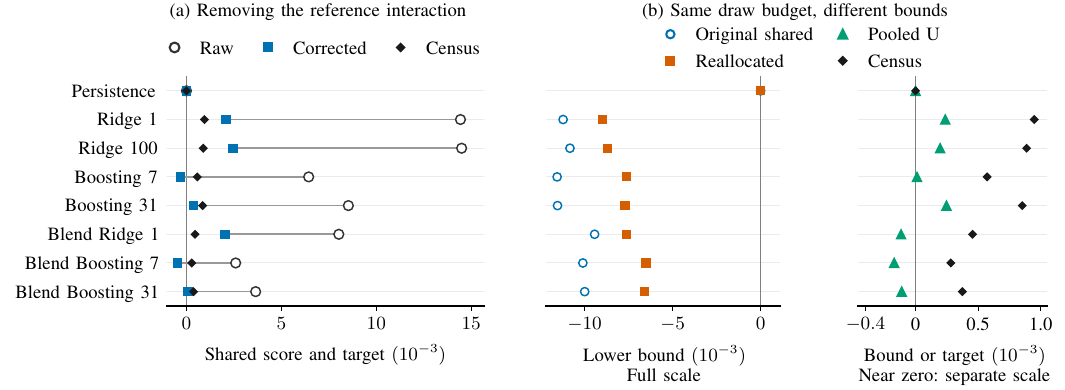}
\caption{Reference correction on the Beijing archive with shared-baseline ranges.
Left: original raw score, corrected center and exact archive target.
Right: hybrid shared bounds and the pooled variance bound; diamonds are exact targets.
Original $(M,n)=(16{,}384,8{,}192)$ and reallocated $(13{,}086,15{,}852)$ designs each use $73{,}728$ draws.
Ridge numbers are penalty strengths; boosting numbers are maximum leaf counts.
Only positive lower bounds certify candidates; Table~\ref{tab:trajectory-decomposition} also reports the direct-triple bound.}
\label{fig:trajectory-comparison}
\end{figure*}

%% file: paper/tab_trajectory_decomposition.tex
\begin{table}[t]\centering\small
\caption{Boosting 31: certificate decomposition.}
\label{tab:trajectory-decomposition}
\setlength{\tabcolsep}{2.5pt}
\begin{tabular}{@{}lrrrr@{}}\toprule
 & Original & Reallocated & Triple & Pooled\\
Quantity ($10^{-3}$) & shared & shared & U & U\\\midrule
Score mean & 8.530 & 10.355 & 0.870 & 0.847\\
Estimated interaction & 8.148 & 8.477 & 0.000 & 0.000\\
Evaluation radius & 4.041 & 4.930 & 0.653 & 0.601\\
Validation radius & 7.886 & 4.660 & 0.000 & 0.000\\
Lower bound & $-11.545$ & $-7.712$ & 0.217 & 0.246\\
\bottomrule
\end{tabular}\par\smallskip
\begin{minipage}{\columnwidth}\footnotesize
All methods use $73{,}728$ draws and the shared-baseline ranges.
Exact target: $0.852$; expected shared score: $9.909$.
Candidate level $0.05/8$; shared validation budget $0.05/16$.
Shared and triple rules use the hybrid; pooled U uses its variance bound.
Lower bound subtracts the interaction and both radii from the mean;
calculations use unrounded values.\end{minipage}
\end{table}

%% file: paper/6_conclusion.tex
Our identities locate induced association in rank-based forecast comparisons.
For fixed maps, reference laws and row sums, separate references or estimate
their interaction. When all trajectories are available, complete U-statistics
give a stronger range bound for the same target.

\paragraph{Limitations}
Observed-panel ranks, additive fits and dependence corrections lack joint calibration.
Aggregate validation needs exact masses and independent pairs; trajectory bounds need
conditionally i.i.d.\ blocks and concern a fitted rank target.
Temporal results require their stated conditions; simulator calibration needs a null law.
Archives establish neither future-loss nor acquisition gains; chronological studies
do not establish a reproducible incremental gating benefit. Earlier-task comparisons are developmental.
Restricted training and historical availability remain qualified.
Non-retention does not prove no predictive value.

\paragraph{Code and data availability}
Code, licensed public data, results and derivations are available.\footnote{\url{https://github.com/MaxwellNi/certifying-forecasting-skill}.}
The restricted application releases only model-level and simulation aggregates;
observations, identifiers, predictions, residuals and loaders remain private.

%% file: paper/appendix.tex
\input{paper/trajectory_design_proof}

Proposition~\ref{prop:fdr} applies the BY theorem~\cite{benjamini2001control},
conditionally on selection when required, then averages. Marginal validity
alone is insufficient: for uniform $p$ and $\alpha<t$,
$\Pr(p\le\alpha\mid p\le t)=\alpha/t$.

\input{paper/reference_rank_main_proof}

\input{paper/feedback_identity_proof}

%% file: paper/trajectory_design_proof.tex
\begin{proof}[Proof of Proposition~\ref{prop:trajectory}]
Conditional independence gives
\[
\mathbb E[p_{sj}q_{tk}\mid\mathsf W_0,\mathcal H]
=m_sn_t+\mathbf1\{j=k\}\operatorname{Cov}(p_{s1},q_{t1}\mid\mathsf W_0,\mathcal H).
\]
Summing coefficients and averaging proves~\eqref{eq:trajectory-identity}.
One Bernoulli$(1/2)$ map pair, with all others constant, isolates
$\Lambda_{st}=1/16$ and proves necessity; sufficiency is immediate.
The independent-copy covariance identity gives $\mathbb E[Q\mid\mathcal H]
=\langle\boldsymbol\Omega,\Lambda\rangle$.
Centered comparisons lie in $[-1/2,1/2]$, hence
$|AB|\le\|\mathbf C\|_1\|\mathbf D\|_1/4$ and
$|Q|\le\|\boldsymbol\Omega\|_1/2$.
Hoeffding and empirical Bernstein each bound the same conditional mean;
splitting their failure budget yields~\eqref{eq:trajectory-hybrid}.
The independent evaluation/validation allocation proves~\eqref{eq:trajectory-bound}.
Monotonicity in $\alpha$ proves the fixed-$\delta$ inversion.

For the comparator, average
$(\sum_s c_s p_{s,ij})(\sum_t d_t q_{t,ik})$ over ordered distinct $i,j,k$.
A scalar pair's six-role average lies in $[-1/12,1/12]$:
ordering the first map's values $x_1<x_2<x_3$ reduces $24$ times this average to
$-2\operatorname{sgn}(y_1-y_3)$ for the second map's values; one tie gives a sum of two signs,
and complete ties give zero.
Triangle inequality gives $W_U$. Classical U-statistic concentration
uses $q$ disjoint triples. Since
$q\ge\lfloor2M/3\rfloor+1>2M/3$ and
$\|c\|_1\|d\|_1\le\|\mathbf C\|_1\|\mathbf D\|_1$,
\[
\frac{h(W_U,q,\alpha)}{h(W_S,M,\alpha-\delta)}
\le\frac13\sqrt{\frac{M\log(1/\alpha)}{q\log\{1/(\alpha-\delta)\}}}
<\frac1{\sqrt6}.
\]
If $W_S=0$, both widths vanish and no radius ratio is asserted.

For the shared baseline, write $S=(x-z)(y-z)$ with uncentered comparisons
$x,y,z\in\{0,1/2,1\}$. Conditioning on $z$ gives $S\in[-1/4,1]$.
For $Q=(u-w)(v-w)/2$, all three differences lie in $[-1,1]$;
the product lies between $w^2-1$ and $(1+|w|)^2$, giving $Q\in[-1/2,2]$.
Let $k(x,y)$ be the scalar six-role kernel above. Then
$U=k(h,Y)-k(h,b)-k(b,Y)+k(b,b)$.
A nonconstant triple of baseline values has \mbox{$k(b,b)=1/12$};
using $|k|\le1/12$ gives $U\in[-1/6,1/3]$.
For a constant baseline, its three terms vanish, leaving $|U|\le1/12$.
All three endpoint pairs are attained by three-point maps, including ties.
\end{proof}

%% file: paper/reference_rank_main_proof.tex
\paragraph{Finite rank references}
\begin{proof}[Proof of Proposition~\ref{prop:reference-main}]
For distinct $i,j,l$, conditional independence gives
$\mathbb E[a_{ij}b_{il}\mid V_i,W_i,Z_i,H]=F_V(V_i)F_W(W_i)$.
Averaging $(a_{ij}-f(Z_i))(b_{il}-g(Z_i))$ gives
$\vartheta+b_H$; the $k$ shared pairs in $C_i$ add $\Gamma/k$.
\end{proof}
\begin{proof}[Proof of Theorem~\ref{thm:reference-main}]
Symmetrize $(a_{12}-m_V(Z_1))(b_{13}-m_W(Z_1))$.
The focal and reference role expectations are $r_Vr_W,h_V(V),h_W(W)$,
each integrating to $\vartheta$, so the first projection is $\psi/3$.
Bounded U-statistic decomposition~\cite{hoeffding1948ustat} gives
$D_j-\vartheta=N^{-1}\sum_i\psi(O_{ji})+R_j$,
$\mathbb ER_j=0$, $\mathbb ER_j^2=O(N^{-2})$.
The scaled mean remainder has variance $O(N^{-1})$; the bounded-projection
CLT applies for the fixed common law and positive limiting variance.
Algebra gives $C_j-D_j=(J_j-Q_j)/(N-1)$, where $J_j,Q_j$ average
shared pairs and distinct triples. Their difference has mean $\Gamma$ and
variance $O(N^{-1})$, so its centered scaled mean has variance $O(N^{-2})$.
This proves both centered limits and $N\operatorname{Var}(A_j)\to\sigma_\psi^2$.
Disjoint-triple symmetrization gives $\mathbb EW_j^4=O(N^{-2})$ for
$W_j=A_j-\mathbb EA_j$. Thus
$\operatorname{Var}[(N/M)\sum_jW_j^2]=O(M^{-1})$ and
$N\bar W^2=O_p(M^{-1})$ imply $N\widehat v_A\to_p\sigma_\psi^2$.
Slutsky gives the boundary limits and
$T_C/[\Gamma\sqrt{MN}/\{(N-1)\sigma_\psi\}]\to_p1$ in the divergent case.
Conditional independence centers the reference-role residuals.
Boundedness makes the remainder and moment bounds uniform.
\end{proof}

\begin{proof}[Proof of Corollary~\ref{cor:estimated-reference}]
Let $\delta_f=f-m_V$, $\delta_g=g-m_W$. The fitted-minus-oracle
group difference, identical for $C,D$, is
\[
\begin{split}
\Delta_j=\frac1N\sum_i\{&-\delta_f(Z_i)(U_{W,i}-m_W(Z_i))\\
&-\delta_g(Z_i)(U_{V,i}-m_V(Z_i))+\delta_f(Z_i)\delta_g(Z_i)\}.
\end{split}
\]
It is an order-two U-statistic with conditional mean $b_H$ and, by bounded
comparisons and clipped fits,
$\operatorname{Var}(\Delta_j\mid H)\le
c_0(\varepsilon_f^2+\varepsilon_g^2)/N$ for an absolute constant $c_0$.
Conditional group independence gives
$\sqrt{MN}(\bar\Delta-b_H)=o_p(1)$; the product-rate condition gives
$\sqrt{MN}b_H=o_p(1)$.
Conditional unbiasedness yields $N\widehat v_\Delta=o_p(1)$, and sample
covariance Cauchy--Schwarz gives
$N|\widehat v_A-\widehat v_A^o|
\le N\widehat v_\Delta+2\sqrt{N\widehat v_\Delta\,N\widehat v_A^o}=o_p(1)$.
The oracle limits and ratio argument transfer.
\end{proof}
\input{paper/reference_certificate_proof}

%% file: paper/reference_certificate_proof.tex
\begin{proof}[Proof of Proposition~\ref{prop:category-certificate}]
Conditional on focal validation categories, rectangular pair comparisons
have means $m_V(c),m_W(c)$; their intervals fail with probability at most
$\delta/2$. The binomial limits, with total failure probability $\delta/2$,
are $u_c=\operatorname{Beta}^{-1}(1-\delta/(2C_Z);
n_c+1,m_{\rm val}-n_c)$ for $n_c<m_{\rm val}$, and one otherwise.
Since $u_c\ge n_c/m_{\rm val}$, $\mathcal P$ is feasible.
On the joint event it contains the true probability vector; bilinearity
gives endpoint extrema and
$\widehat B_-\le b_H\le\widehat B_+\le\widehat B_{\rm abs}$.

Here $\mathcal H_{\rm val}$ denotes independent training and validation.
Let $\sigma^2,R$ be the disjoint-triple kernel variance and range.
Permutation averaging and exponential convexity transfer the independent
$J$-kernel Hoeffding and Bernstein moment bounds to
$\bar D$~\cite{hoeffding1963}. The upper-deviation radii
$R\sqrt{x/(2J)}$ and $\sigma\sqrt{2x/J}+Rx/(3J)$ each fail with
probability at most $e^{-x}$.
Theorem~10 of~\cite{maurer2009empirical} gives
$\sigma\le s+R\sqrt{2x/(J-1)}$ except with probability $e^{-x}$.
Conditional on $\mathcal H_{\rm val}$, the smaller true-variance radius
is fixed. Union with the variance event proves
Eq.~\eqref{eq:variance-radius} at failure $2e^{-x}$, without independence
of $s,\bar D$. Since
$\mathbb E[\bar D\mid\mathcal H_{\rm val}]=\vartheta+b_H$,
adding validation failure $\delta$ proves the bound.

For fixed $\delta$, the lower bound increases with $\alpha$, so monotone
inversion gives super-uniform $p$-values, including $p=1$ when the
corrected score does not exceed the learning allowance.
Permutation averaging also transfers the negative-exponent Hoeffding bound.
On the joint validation event, with evaluation probability at least $1-\eta$,
the lower certificate is at least
$\vartheta-W-R[\sqrt{\log(1/(\alpha-\delta))}+
\sqrt{\log(1/\eta)}]/\sqrt{2J}$,
proving Eq.~\eqref{eq:aggregate-power-cost}.
\end{proof}

\input{paper/aggregate_validation_proof}

%% file: paper/aggregate_validation_proof.tex
\begin{proof}[Proof of Lemma~\ref{lem:aggregate-allowance}]
If $\mathcal S$ is empty, the deterministic corner bound suffices.
Otherwise, condition only on the focal validation categories, and write $a_c=m_V(c)-f(c)$,
$b_c=m_W(c)-g(c)$ and $e_{Ac}=\widehat a_c-a_c$,
$e_{Bc}=\widehat b_c-b_c$. The errors are independent across streams
and categories, with sub-Gaussian proxies $1/(4n_{Ac})$ and
$1/(4n_{Bc})$ by Hoeffding's lemma. Over $\mathcal S$,
\[
\begin{split}
\sum_cp_c(a_cb_c-\widehat a_c\widehat b_c)
={}&-\sum_cp_c\widehat b_ce_{Ac}-\sum_cp_c\widehat a_ce_{Bc}\\
 &+\sum_cp_ce_{Ac}e_{Bc}.
\end{split}
\]
Conditioning on the opposite stream gives upper-tail bounds $r_A,r_B$,
each failing with probability at most $e^{-x}$. For independent centered
sub-Gaussian $X,Y$ with proxies $s_X^2,s_Y^2$, condition on $X$ and
integrate a standard normal variable to obtain
$\mathbb E e^{tXY}\le(1-t^2s_X^2s_Y^2)^{-1/2}$ for
$|t|s_Xs_Y<1$. Thus $R_2=\sum_cp_ce_{Ac}e_{Bc}$ obeys
\[
 \log\mathbb E e^{tR_2}\le
 \frac{t^2\sum_cd_c^2}{2(1-t\max_cd_c)},\qquad
 0<t<1/\max_cd_c.
\]
Here $-\log(1-u^2)\le u^2/(1-|u|)$. Chernoff optimization gives
$\Pr(R_2>r_{AB})\le e^{-x}$. A union bound over the three tails gives
$3e^{-x}=\delta$ without requiring their events to be independent.
Missing categories contribute at most their deterministic corner bound;
conditioning fixes the missing set. Remove conditioning to finish.
\end{proof}

%% file: paper/feedback_identity_proof.tex
\begin{proof}[Proof of Proposition~\ref{prop:feedback-main}]
The entity effect cancels: $y_r-\bar y_O=\varepsilon_r-\bar\varepsilon_O$.
Shock orthogonality removes the other terms. \mbox{For $r\in I$,}
\[
\begin{aligned}
\frac{\mathbb E[(x_r-\bar x_O)(y_r-\bar y_O)]}{\sigma^2}
={}&-\frac{\sum_{j\in O}w_{rj}}{m}\\
&-\frac{\sum_{s\in O}w_{sr}}{m}
 +\frac{\sum_{s,j\in O}w_{sj}}{m^2}.
\end{aligned}
\]
The first two sums reduce to $E$ and $L$, giving Eq.~\eqref{eq:feedback-main}.
In the split product, forecast shocks precede $r$ and outcome shocks
occur at or after $r$, so its expectation is zero.
For nonlinear forecasts, set
$\mathcal G_s=\sigma(\alpha,\eta_1,\ldots,\eta_T,\varepsilon_1,\ldots,\varepsilon_{s-1})$.
For $s\ge r$, the forecast residual is $\mathcal G_s$-measurable, giving
$\mathbb E[(x_r-\bar x_E)\varepsilon_s]=0$.
Apply this at $r$ and average over $L$.
\end{proof}

\begin{proof}[Proof of Corollary~\ref{cor:directional}]
Each $S_g$ has mean zero. Bounded fourth moments give Lyapunov
ratio $O(G^{-1})$; independence gives
$G^{-1}\sum_g\{S_g^2-\mathbb ES_g^2\}\to_p0$ and $\bar S\to_p0$.
The sample variance converges to $v$; studentization completes the proof.
\end{proof}

\paragraph{Fixed peers}
Conditional on entity attributes, channel independence gives
$\mathbb E[U_{V,i}U_{W,i}]=\bar p_i\bar q_i$ and
$\mathbb E[a_{ij}b_{ij}]=p_{ij}q_{ij}$, where
$\bar p_i=k^{-1}\sum_{j\ne i}p_{ij}$ and similarly for $\bar q_i$.
Exact centering gives $\mathbb EC_i=0$; Eq.~\eqref{eq:reference-correction-main}
then gives Eq.~\eqref{eq:fixed-peers-main}.

%% file: ref.bib
@article{chernozhukov2018dml,
  author  = {Chernozhukov, Victor and Chetverikov, Denis and Demirer, Mert and Duflo, Esther and Hansen, Christian and Newey, Whitney and Robins, James},
  title   = {Double/debiased machine learning for treatment and structural parameters},
  journal = {The Econometrics Journal},
  year    = {2018},
  volume  = {21},
  number  = {1},
  pages   = {C1--C68}
}

@article{shah2020hardness,
  author  = {Shah, Rajen D. and Peters, Jonas},
  title   = {The hardness of conditional independence testing and the generalised covariance measure},
  journal = {The Annals of Statistics},
  year    = {2020},
  volume  = {48},
  number  = {3},
  pages   = {1514--1538}
}

@inproceedings{zhang2011kernel,
  title={Kernel-based conditional independence test and application in causal discovery},
  author={Zhang, Kun and Peters, Jonas and Janzing, Dominik and Sch{\"o}lkopf, Bernhard},
  booktitle={Uncertainty in Artificial Intelligence}, pages={804--813}, year={2011}
}

@article{benjamini2001control,
  title={The control of the false discovery rate in multiple testing under dependency},
  author={Benjamini, Yoav and Yekutieli, Daniel},
  journal={The Annals of Statistics}, volume={29}, number={4}, pages={1165--1188}, year={2001}
}

@inproceedings{pan2026eccit,
  title={Empirically Calibrated Conditional Independence Tests},
  author={Pan, Milleno and de Mathelin, Antoine and Tansey, Wesley},
  booktitle={International Conference on Artificial Intelligence and Statistics},
  series={Proceedings of Machine Learning Research},
  volume={300},
  pages={5194--5202},
  year={2026}
}

@article{chiang2022multiway,
  title={Multiway cluster robust double/debiased machine learning},
  author={Chiang, Harold D. and Kato, Kengo and Ma, Yukun and Sasaki, Yuya},
  journal={Journal of Business \& Economic Statistics},
  volume={40},
  number={3},
  pages={1046--1056},
  year={2022}
}

@article{harper2015movielens,
  author  = {Harper, F. Maxwell and Konstan, Joseph A.},
  title   = {The {MovieLens} datasets: History and context},
  journal = {ACM Transactions on Interactive Intelligent Systems},
  year    = {2015},
  volume  = {5},
  number  = {4},
  pages   = {19:1--19:19}
}

@article{makridakis2022m5,
  author  = {Makridakis, Spyros and Spiliotis, Evangelos and Assimakopoulos, Vassilios},
  title   = {{M5} accuracy competition: Results, findings, and conclusions},
  journal = {International Journal of Forecasting},
  year    = {2022},
  volume  = {38},
  number  = {4},
  pages   = {1346--1364}
}

@misc{trindade2015electricity,
  author       = {Trindade, Artur},
  title        = {{ElectricityLoadDiagrams20112014}},
  howpublished = {UCI Machine Learning Repository},
  year         = {2015},
  note         = {doi: 10.24432/C58C86}
}

@article{chen2022open,
  author  = {Chen, Andrew Y. and Zimmermann, Tom},
  title   = {Open source cross-sectional asset pricing},
  journal = {Critical Finance Review},
  year    = {2022},
  volume  = {11},
  number  = {2},
  pages   = {207--264}
}

@article{neweywest1987,
  title={A simple, positive semi-definite, heteroskedasticity and autocorrelation consistent covariance matrix},
  author={Newey, Whitney K. and West, Kenneth D.},
  journal={Econometrica}, volume={55}, number={3}, pages={703--708}, year={1987}
}

@article{hoeffding1963,
  author={Hoeffding, Wassily},
  title={Probability Inequalities for Sums of Bounded Random Variables},
  journal={Journal of the American Statistical Association},
  volume={58}, number={301}, pages={13--30}, year={1963},
  doi={10.1080/01621459.1963.10500830}
}

@article{clopper1934use,
  author={Clopper, C. J. and Pearson, E. S.},
  title={The Use of Confidence or Fiducial Limits Illustrated in the Case of the Binomial},
  journal={Biometrika}, volume={26}, number={4}, pages={404--413}, year={1934},
  doi={10.1093/biomet/26.4.404}
}

@article{hoeffding1948ustat,
 author={Hoeffding, Wassily},
 title={A Class of Statistics with Asymptotically Normal Distribution},
 journal={The Annals of Mathematical Statistics},
 volume={19}, number={3}, pages={293--325}, year={1948},
 doi={10.1214/aoms/1177730196}
}

@article{petersen2021partial,
  author = {Petersen, Lasse and Hansen, Niels Richard},
  title = {Testing Conditional Independence via Quantile Regression Based Partial Copulas},
  journal = {Journal of Machine Learning Research},
  year = {2021},
  volume = {22},
  number = {70},
  pages = {1--47}
}

@article{arellano1995another,
  title={Another look at the instrumental variable estimation of error-components models},
  author={Arellano, Manuel and Bover, Olympia},
  journal={Journal of Econometrics},
  volume={68},
  number={1},
  pages={29--51},
  year={1995},
  doi={10.1016/0304-4076(94)01642-D}
}

@article{so1999recursive,
  title={Recursive mean adjustment in time-series inferences},
  author={So, Beong Soo and Shin, Dong Wan},
  journal={Statistics \& Probability Letters},
  volume={43},
  number={1},
  pages={65--73},
  year={1999},
  doi={10.1016/S0167-7152(98)00247-8}
}

@article{leung2018testing,
  author  = {Leung, Dennis and Drton, Mathias},
  title   = {Testing independence in high dimensions with sums of rank correlations},
  journal = {The Annals of Statistics},
  year    = {2018},
  volume  = {46},
  number  = {1},
  pages   = {280--307},
  doi     = {10.1214/17-AOS1550}
}

@article{fair1990comparing,
 author={Fair, Ray C. and Shiller, Robert J.},
 title={Comparing Information in Forecasts from Econometric Models},
 journal={American Economic Review}, volume={80}, number={3}, pages={375--389}, year={1990}
}

@article{giacomini2006tests,
 author={Giacomini, Raffaella and White, Halbert},
 title={Tests of Conditional Predictive Ability},
 journal={Econometrica}, volume={74}, number={6}, pages={1545--1578}, year={2006},
 doi={10.1111/j.1468-0262.2006.00718.x}
}

@inproceedings{maurer2009empirical,
  author = {Andreas Maurer and Massimiliano Pontil},
  title = {Empirical {Bernstein} Bounds and Sample Variance Penalization},
  booktitle = {Proceedings of the 22nd Annual Conference on Learning Theory},
  year = {2009}
}

@inproceedings{peel2010empirical,
  author = {Thomas Peel and Sandrine Anthoine and Liva Ralaivola},
  title = {Empirical {Bernstein} Inequalities for {U}-Statistics},
  booktitle = {Advances in Neural Information Processing Systems},
  volume = {23},
  year = {2010}
}

@article{schechtman1987gini,
 author={Schechtman, Edna and Yitzhaki, Shlomo},
 title={A Measure of Association Based on {Gini}'s Mean Difference},
 journal={Communications in Statistics: Theory and Methods},
 volume={16}, number={1}, pages={207--231}, year={1987},
 doi={10.1080/03610928708829359}
}

@article{waudby2024betting,
 author={Waudby-Smith, Ian and Ramdas, Aaditya},
 title={Estimating means of bounded random variables by betting},
 journal={Journal of the Royal Statistical Society Series B},
 volume={86}, number={1}, pages={1--27}, year={2024},
 doi={10.1093/jrsssb/qkad009}
}

@article{kim2025semi,
 author={Kim, Ilmun and Wasserman, Larry and Balakrishnan, Sivaraman and Neykov, Matey},
 title={Semi-Supervised {U}-statistics},
 journal={The Annals of Statistics}, volume={53}, number={6},
 pages={2488--2515}, year={2025}, doi={10.1214/25-AOS2550}
}

@misc{candanedo2017appliances,
 author={Candanedo, Luis}, title={Appliances Energy Prediction},
 howpublished={UCI Machine Learning Repository}, year={2017},
 doi={10.24432/C5VC8G}
}

@misc{hogue2019metro,
 author={Hogue, John}, title={Metro Interstate Traffic Volume},
 howpublished={UCI Machine Learning Repository}, year={2019},
 doi={10.24432/C5X60B}
}

@article{liu2024falsification,
 author={Liu, Lin and Mukherjee, Rajarshi and Robins, James M.},
 title={Assumption-lean falsification tests of rate double-robustness of double-machine-learning estimators},
 journal={Journal of Econometrics}, volume={240}, number={2}, pages={105500},
 year={2024}, doi={10.1016/j.jeconom.2023.105500}
}

@article{robins2017higher,
 author={Robins, James M. and Li, Lingling and Mukherjee, Rajarshi and Tchetgen Tchetgen, Eric and van der Vaart, Aad},
 title={Minimax estimation of a functional on a structured high-dimensional model},
 journal={The Annals of Statistics}, volume={45}, number={5}, pages={1951--1987}, year={2017},
 doi={10.1214/16-AOS1515}
}

@misc{mcgrath2026tuning,
  author = {Sean McGrath and Rajarshi Mukherjee},
  title = {Nuisance Function Tuning and Sample Splitting for Optimally Estimating a Doubly Robust Functional},
  year = {2026},
  note = {arXiv:2212.14857v5, to appear in The Annals of Statistics}
}

@misc{news2018data,
 author={Lu{\'i}s Torgo and Nuno Moniz},
 title={News Popularity in Multiple Social Media Platforms},
 year={2018},
 howpublished={UCI Machine Learning Repository},
 doi={10.24432/C5H029}
}

@IEEEtranBSTCTL{fullauthornames,
 CTLdash_repeated_names = {no}
}

@article{moscovich2022preprocessing,
 author={Moscovich, Amit and Rosset, Saharon},
 title={On the Cross-Validation Bias due to Unsupervised Preprocessing},
 journal={Journal of the Royal Statistical Society Series B: Statistical Methodology},
 volume={84}, number={4}, pages={1474--1502}, year={2022},
 doi={10.1111/rssb.12537}
}

@misc{chen2015beijing,
  author={Song Chen},
  title={{Beijing PM2.5}},
  year={2015},
  howpublished={UCI Machine Learning Repository},
  doi={10.24432/C5JS49}
}

@article{hewamalage2023pitfalls,
  author = {Hansika Hewamalage and Klaus Ackermann and Christoph Bergmeir},
  title = {Forecast evaluation for data scientists: common pitfalls and best practices},
  journal = {Data Mining and Knowledge Discovery},
  volume = {37},
  number = {2},
  pages = {788--832},
  year = {2023},
  doi = {10.1007/s10618-022-00894-5}
}

@article{clark2007approximately,
  author = {Todd E. Clark and Kenneth D. West},
  title = {Approximately normal tests for equal predictive accuracy in nested models},
  journal = {Journal of Econometrics},
  volume = {138},
  number = {1},
  pages = {291--311},
  year = {2007},
  doi = {10.1016/j.jeconom.2006.05.023}
}

@article{choe2024comparing,
  author = {Yo Joong Choe and Aaditya Ramdas},
  title = {Comparing Sequential Forecasters},
  journal = {Operations Research},
  volume = {72},
  number = {4},
  pages = {1368--1387},
  year = {2024},
  doi = {10.1287/opre.2021.0792}
}
